\documentclass[journal]{IEEEtran}
\usepackage{amsmath,amsfonts,bm,mathrsfs,amssymb,gensymb,multirow,color}
\usepackage{algorithmic}
\usepackage{algorithm}
\usepackage{booktabs}
\usepackage{array}
\usepackage[caption=false,font=normalsize,labelfont=sf,textfont=sf]{subfig}
\usepackage{textcomp}
\usepackage{stfloats}
\usepackage{url}
\usepackage{verbatim}
\usepackage{graphicx}
\usepackage[switch,pagewise]{lineno} 
\usepackage{mathtools}
\usepackage{float}

\graphicspath{{Fig/}}

\usepackage{ntheorem}
\newtheorem{theorem}{Theorem}

\newtheorem*{proof}{Proof}
\newtheorem{proposition}{Proposition}

\def\BibTeX{{\rm B\kern-.05em{\sc i\kern-.025em b}\kern-.08em
    T\kern-.1667em\lower.7ex\hbox{E}\kern-.125emX}}
\usepackage{balance}
\begin{document}

\title{On the Convergence of Belief Propagation for Multipath Data Association in Target Tracking}

\author{Kuilong Yang, Zengfu Wang*, Hua Lan 
\thanks{All authors are with the School of Automation, Northwestern Polytechnical University, Xi'an 710072, China.
This work was supported in part by the National Natural Science Foundation of China~(Grants No.~62473317, U21B2008, 62371398).
\emph{(Corresponding author: Zengfu Wang.)}}
}

\markboth{Journal of \LaTeX\ Class Files,~Vol.~18, No.~9, September~2020}%
{How to Use the IEEEtran \LaTeX \ Templates}

\maketitle

\begin{abstract}
Belief propagation (BP) is widely used for data
association (DA) in target tracking.
Existing convergence analyses of BP for DA address
only the two-way correspondence between targets and
measurements, where each target generates at most one
measurement per scan.
Multipath DA (MPDA) allows a single target to produce
multiple measurements via distinct propagation paths,
creating a three-way correspondence among targets,
paths, and measurements, for which a complete convergence
proof has not yet been provided.
We provide such a proof for the BP updates in MPDA,
establishing convergence to a unique fixed point.
Simulations illustrate the convergence behavior of BP in MPDA and demonstrate a favorable accuracy--efficiency trade-off relative to both single-scan and two-scan variants of the multiple-detection multiple-hypothesis tracker.
\end{abstract}


\section{Introduction}
Multipath data association (MPDA) is a critical challenge in multipath detection systems (MDS), including skywave over-the-horizon radar (OTHR) \cite{Headrick1974High, Pulford1998, guo2021othr}, passive radars \cite{Tharmarasa2012Multitarget}, urban radar networks \cite{Zhou2016Multiple, Li2014Simultaneous}, and wireless simultaneous localization and mapping \cite{LeitingerMeyerMultipathSLAM, Gao2024SLAM}. 
In MDS, a single point target may generate multiple measurements via different propagation paths. 
The unknown correspondence between targets, propagation paths, and measurements has to be resolved.
Traditional data association (DA) methods, 
including multipath probabilistic data association (PDA) \cite{Pulford1998}, multiple-detection joint PDA (MD-JPDA) \cite{MDJPDA}, and multiple-detection multiple hypothesis tracking (MD-MHT) \cite{MHT2013Sathyan}, may suffer from combinatorial explosion in joint target-measurement-path correspondences, or information loss from probabilistic approximations.


Belief propagation (BP) provides an efficient and scalable
inference framework for MPDA via factor
graphs~\cite{LeitingerMeyerMultipathSLAM, Lan2020Joint,
Lan2021Measurement}.
Since BP is an iterative algorithm, establishing its
convergence is important for reliable deployment.
The prior convergence of BP for two-way DA---where each target
generates at most one measurement per scan---was proved
in~\cite{Williams2014Approximate} using the contraction
mapping theorem.
While~\cite{Lan2020Joint} observed a 
convergence argument for BP in MPDA by treating each 
(target, path) pair as a pseudo-target within the framework
of~\cite{Williams2014Approximate}, a complete convergence 
proof tailored to MPDA has not yet been demonstrated; a detailed discussion is
given in Section~\ref{sec:standardDA}.

In this correspondence, we prove that, for the MPDA
formulation of~\cite{Lan2020Joint}, every execution of
the sum-product BP algorithm converges to a unique
fixed point.
The proof builds on the contraction
lemmas of~\cite{Williams2014Approximate}, which apply
to the MPDA message updates after a straightforward
algebraic reformulation, and establishes convergence
via the Banach fixed-point
theorem~\cite{Banach1922}.
Simulation results demonstrate the convergence behavior of
BP and its favorable tracking accuracy--efficiency trade-off
relative to MD-MHT~\cite{MHT2013Sathyan}.

\section{BP Algorithm For MPDA}
\subsection{Problem Description}\label{Sec:problem}
We follow the description of the MPDA problem in the context of multiple target tracking~(MTT) in \cite{Lan2020Joint}. 
For any positive integer $N$, 
let $[N]$ represent the set $\{1,2,\ldots,N\}$.
Let $x_{i,k}\in \mathbb{R}^{n_x}$ be the kinematic state of target $i$ at time $k$. The discrete-time dynamics for $\mathfrak{X}_k$ independent targets are given by $x_{i,k+1} = f_{i,k}(x_{i,k}) + u_{i,k}$, $i \in [\mathfrak{X}_k]$, where $f_{i,k}(\cdot)$ is the known state transition function and $u_{i,k} \sim \mathcal{N}(\mathbf{0}, Q_{i,k})$ is the zero-mean Gaussian process noise with covariance $Q_{i,k}$. 
At time $k$, the system receives $\mathfrak{Y}_k$ measurements $y_{j,k} \in \mathbb{R}^{n_y}$, $j \in [\mathfrak{Y}_k]$. 
A measurement $j$ originating from target $i$ via an unknown propagation path $\tau \in [\mathfrak{M}_k]$ with detection probability $p_{\text{d}}^{\tau} \in (0,1)$ is modeled as $y_{j,k} = h_{\tau,k}(x_{i,k}) + v_{\tau,k}$, where $h_{\tau,k}(\cdot)$ is the measurement function and $v_{\tau,k} \sim \mathcal{N}(\mathbf{0}, R_{\tau,k})$ is the zero-mean Gaussian measurement noise with covariance $R_{\tau,k}$, and $\mathfrak{M}_k$ is the total number of propagation paths. 
Clutter is uniformly distributed over the surveillance volume $V_k$ with clutter spatial density $\lambda$.
Since each received measurement may be clutter or may originate from an unknown target through an unknown propagation path, the correspondence among targets, propagation paths, and measurements has to be resolved.
Let $X_k=\{ x_{i,k} \} _{i=1}^{\mathfrak{X}_k}$ and $Y_k = \{y_{j,k}\}_{j=1}^{\mathfrak{Y}_k}$ denote the sets of target kinematic states and measurements at time $k$, respectively.
Let $Y_{1:k} = \{Y_1, Y_2, \ldots, Y_k\}$.

Let $A_k \triangleq
(a_k^{i,j,\tau})_{i\in[\mathfrak{X}_k],\,
j\in\{0\}\cup[\mathfrak{Y}_k],\,
\tau\in[\mathfrak{M}_k]}
\cup\,
(a_k^{0,j})_{j\in[\mathfrak{Y}_k]}$
denote an MPDA event at time~$k$.
Here, representing the association variable, $a_{k}^{i,j,\tau}$ and $a_k^{0,j}$ take values in $\{ 0,1 \}$ and signify an association event between target, measurement, and path. Specifically, $a_{k}^{i,j,\tau}$, $i>0$, $j>0$ indicates that measurement $j$ originates from target $i$ via path $\tau$;
$a_{k}^{i,0,\tau}$, $i>0$, signifies that target $i$ is not detected via path $\tau$;
$a_k^{0,j}$, $j>0$, denotes measurement $j$ is clutter, where the index $\tau$ is omitted since clutter does not have an identifiable propagation path.
An MPDA event $A_k$ is feasible if it satisfies the following two constraints \cite{Lan2020Joint,Lan2021Measurement},
\begin{align}
    &\textstyle{\sum_{j = 0}^{\mathfrak{Y}_k}}{a_{k}^{i,j,\tau} = 1}, \quad \forall i \in  [\mathfrak{X}_k], \forall \tau \in  [\mathfrak{M}_k],\label{Eq:FrameCon1}  \\
      &\textstyle{\sum_{i = 1}^{\mathfrak{X}_k}{\sum_{\tau = 1}^{\mathfrak{M}_k}}{a_{k}^{i,j,\tau} + a_k^{0,j}}} = 1,\quad \forall j \in  [\mathfrak{Y}_k].  \label{Eq:FrameCon2}
\end{align}
Let $\mathcal{A}_k$ denote the set of all feasible MPDA 
events satisfying constraints~\eqref{Eq:FrameCon1}--\eqref{Eq:FrameCon2} at time $k$.

We focus on the BP-based MPDA inference module in the
joint detection and tracking based on variational Bayes
(JDT-VB) framework of~\cite{Lan2020Joint}.
In that
framework, the joint detection and tracking algorithm
consists of three coupled modules: MPDA inference (Module~3), target existence state
estimation (Module~2), and target kinematic state
estimation (Module~1), which are updated within an
outer VB loop. Given the current posterior
approximations of target kinematic states and target
existence states, Module~3 constructs the variational
parameters associated with the MPDA variables and
approximately evaluates the corresponding marginal
association probabilities via BP.

More specifically, for each time index~$k$ within an
outer VB iteration, Module~3 operates on the MPDA
event~$A_k$.
Following~\cite[(42)]{Lan2020Joint}, conditioned
on the current variational parameter
vector~$\chi_{k}$, the probability mass function (PMF) of 
the MPDA event takes the form,
\begin{equation}\label{Eq:AssociationPMF}
  q(A_k;\chi_{k})
  = \frac{1}{\mathcal{Z}_k(\chi_{k})}
    \exp\!\bigl(\chi_{k}^{\mathsf{T}}A_k\bigr)
    \,\mathbf{1}_{\mathcal{A}_k}(A_k),
\end{equation}
where~$\chi_{k}^{\mathsf{T}}A_k$ denotes the inner
product between the parameter vector and the stacked
binary association variables,
$\mathcal{Z}_k(\chi_{k})$ is the normalizing
constant,
and~$\mathbf{1}_{\mathcal{A}_k}(A_k)$ restricts~$A_k$
to the feasible set defined
by~\eqref{Eq:FrameCon1}--\eqref{Eq:FrameCon2}.
Following~\cite[(43)]{Lan2020Joint}, the elements
of~$\chi_{k}$ comprise
$\chi_{k}^{i,j,\tau}$ for each
target--measurement--path triplet with
$i\in[\mathfrak{X}_k]$, $j\in[\mathfrak{Y}_k]$, and
$\tau\in[\mathfrak{M}_k]$;
$\chi_{k}^{i,0,\tau}$ for each missed-detection
event with $i\in[\mathfrak{X}_k]$ and
$\tau\in[\mathfrak{M}_k]$; and
$\chi_{k}^{0,j}$ for each clutter event with
$j\in[\mathfrak{Y}_k]$.

In the JDT-VB procedure, $\chi_{k}$ is recalculated
at each outer VB iteration using updated estimates from
Module~1 and Module~2. However, once Module~3 is
entered, $\chi_{k}$ remains fixed throughout the
inner BP execution. 
The convergence result established below applies to each individual 
execution of Module~3, but does not address the convergence of the 
outer VB loop.

\subsection{Factor Graph Modeling}

Following~\cite[(44)]{Lan2020Joint}, the
PMF~\eqref{Eq:AssociationPMF} factorizes as
\begin{equation}\label{Eq:AssociationFactorization}
\begin{aligned}
\!\!\! q(A_k;\chi_{k}) \propto{}&
\prod_{i=1}^{\mathfrak{X}_k}
\prod_{\tau=1}^{\mathfrak{M}_k}
f_{\mathcal{T}_k}(i,\tau)\;
\prod_{j=1}^{\mathfrak{Y}_k}
f_{\mathcal{M}_k}(j) \\
&\times
\prod_{i=1}^{\mathfrak{X}_k}
\prod_{\tau=1}^{\mathfrak{M}_k}
\prod_{j=0}^{\mathfrak{Y}_k}
f_{\mathcal{E}_k}^{i,j,\tau}(a_k^{i,j,\tau})\;
\prod_{j=1}^{\mathfrak{Y}_k}
f_{\mathcal{E}_k}^{0,j}(a_k^{0,j}),
\end{aligned}
\end{equation}
where $f_{\mathcal{T}_k}(i,\tau)$
and $f_{\mathcal{M}_k}(j)$ enforce
\eqref{Eq:FrameCon1} and~\eqref{Eq:FrameCon2},
respectively,
\begin{align}
  f_{\mathcal{T}_k}(i,\tau)
  &= \mathbf{1}\!\Bigl(
       \textstyle\sum_{j=0}^{\mathfrak{Y}_k}
       a_k^{i,j,\tau}=1\Bigr),
  \label{Eq:fT}\\
  f_{\mathcal{M}_k}(j)
  &= \mathbf{1}\!\Bigl(
       \textstyle\sum_{i=1}^{\mathfrak{X}_k}
       \sum_{\tau=1}^{\mathfrak{M}_k}
       a_k^{i,j,\tau}+a_k^{0,j}=1\Bigr),
  \label{Eq:fM}
\end{align}
and $f_{\mathcal{E}_k}^{i,j,\tau}$,
$f_{\mathcal{E}_k}^{0,j}$ encode the local evidence,
\begin{align}
  f_{\mathcal{E}_k}^{i,j,\tau}(a_k^{i,j,\tau})
  &= \exp\!\bigl(
       \chi_{k}^{i,j,\tau}\,a_k^{i,j,\tau}\bigr),
  \label{Eq:fEtarget}\\
  f_{\mathcal{E}_k}^{0,j}(a_k^{0,j})
  &= \exp\!\bigl(
       \chi_{k}^{0,j}\,a_k^{0,j}\bigr).
  \label{Eq:fEclutter}
\end{align}

Fig.~\ref{dataAssociation_Multipath2} illustrates the
subgraph of the factorization~\eqref{Eq:AssociationFactorization} formed by
the DA variables~$a_k^{i,j,\tau}$ and the constraint
factors~$f_{\mathcal{T}_k}(i,\tau)$
and~$f_{\mathcal{M}_k}(j)$; here we focus on the BP
message updates required for the convergence analysis.
The circular nodes represent the DA variables, while the
rectangular nodes represent the constraint factors.
Variables connected to a
common~$f_{\mathcal{T}_k}(i,\tau)$ are grouped within
the blue shaded areas, and those connected to a
common~$f_{\mathcal{M}_k}(j)$ within the green shaded
areas.
In total, the factor graph comprises
$\mathfrak{X}_k\mathfrak{Y}_k\mathfrak{M}_k
+\mathfrak{X}_k\mathfrak{M}_k +\mathfrak{Y}_k$ variable
nodes, interconnected by
$\mathfrak{X}_k\mathfrak{M}_k$
factors~$f_{\mathcal{T}_k}$ and $\mathfrak{Y}_k$
factors~$f_{\mathcal{M}_k}$.
For the complete factor graph for joint detection and
tracking in MDS, the reader is referred
to~\cite{Lan2020Joint}.

\begin{figure}[!b]
\centering
\vspace{-4pt}
\includegraphics[width=0.34\textwidth]{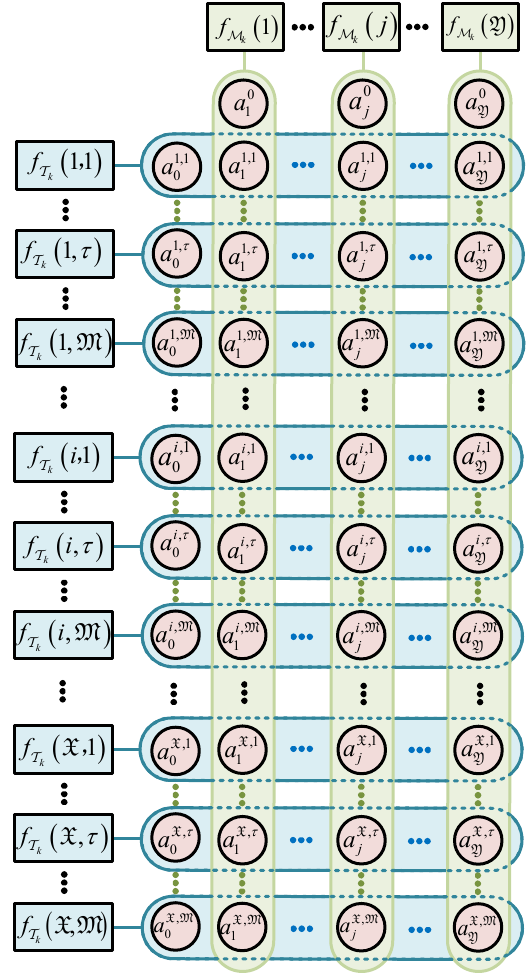}
\caption{The subgraph representing the MPDA constraints. For notational simplicity, the time index $k$ in the variable nodes is omitted, i.e., $a_j^{i,\tau} = a_{k}^{i,j,\tau}$ and $a_j^0 = a_k^{0,j}$. 
Each variable node is also connected to its corresponding evidence factor, i.e., 
$a_k^{i,j,\tau}$ to $f_{\mathcal{E}_k}^{i,j,\tau}$ and $a_k^{0,j}$ to $f_{\mathcal{E}_k}^{0,j}$; 
these evidence factors are omitted here for visual clarity.}
\label{dataAssociation_Multipath2}
\vspace{-8pt}
\end{figure}

\subsection{Simplified BP Algorithm for MPDA Inference}
The BP messages for MPDA inference are derived
in~\cite[(47)--(55)]{Lan2020Joint}. 
We summarize
the simplified scalar ratio form used in this work.

For each evidence factor, the fixed message is
$\overline{\mu}_{\mathcal{E}_k}^{i,j,\tau}
 \triangleq \exp(\chi_{k}^{i,j,\tau})
 \in\mathbb{R}_{++}$
and
$\overline{\mu}_{\mathcal{E}_k}^{0,j}
 \triangleq \exp(\chi_{k}^{0,j})
 \in\mathbb{R}_{++}$,
passed from the evidence factors
$f_{\mathcal{E}_k}^{i,j,\tau}$ and
$f_{\mathcal{E}_k}^{0,j}$ to the corresponding variable
nodes. These messages remain fixed throughout one
execution of the BP iteration.

The iterative constraint messages
$\overline{\mu}_{\mathcal{T}_k}^{i,j,\tau},\,
 \overline{\mu}_{\mathcal{M}_k}^{i,j,\tau}
 \in\mathbb{R}_{++}$,
passed from the constraint factors
$f_{\mathcal{T}_k}(i,\tau)$ and $f_{\mathcal{M}_k}(j)$
to the variable node~$a_k^{i,j,\tau}$, enforce
\eqref{Eq:FrameCon1} and~\eqref{Eq:FrameCon2}.
For all $i\in[\mathfrak{X}_k]$,
$j\in[\mathfrak{Y}_k]$, and
$\tau\in[\mathfrak{M}_k]$, their synchronous updates
are,
\begin{align}
  \overline{\mu}_{\mathcal{T}_k}^{i,j,\tau}
  &= \frac{1}{\overline{\mu}_{\mathcal{E}_k}^{i,0,\tau}
     +\sum_{\substack{j_1=1,j_1\neq j}}^{\mathfrak{Y}_k}
       \overline{\mu}_{\mathcal{E}_k}^{i,j_1,\tau}\,
       \overline{\mu}_{\mathcal{M}_k}^{i,j_1,\tau}},
  \label{Eq:muTk}\\[4pt]
  \overline{\mu}_{\mathcal{M}_k}^{i,j,\tau}
  &= \frac{1}{\overline{\mu}_{\mathcal{E}_k}^{0,j}
     +\sum_{\substack{i_1=1,\tau_1=1\\ (i_1,\tau_1)\neq(i,\tau)}}^{\mathfrak{X}_k,\mathfrak{M}_k}
       \overline{\mu}_{\mathcal{E}_k}^{i_1,j,\tau_1}\,
       \overline{\mu}_{\mathcal{T}_k}^{i_1,j,\tau_1}},
  \label{Eq:muMk}
\end{align}
Upon convergence, the approximate marginal probability
of the association event $a_k^{i,j,\tau}=1$ is,
\begin{equation}\label{Eq:belief}
  b_A(a_k^{i,j,\tau}\!=\!1)
  = \frac{\overline{\mu}_{\mathcal{E}_k}^{i,j,\tau}\,
          \overline{\mu}_{\mathcal{T}_k}^{i,j,\tau}\,
          \overline{\mu}_{\mathcal{M}_k}^{i,j,\tau}}
         {1+\overline{\mu}_{\mathcal{E}_k}^{i,j,\tau}\,
            \overline{\mu}_{\mathcal{T}_k}^{i,j,\tau}\,
            \overline{\mu}_{\mathcal{M}_k}^{i,j,\tau}}.
\end{equation}

Let
$\boldsymbol{\mu}_{\mathcal{E}_k}$ denote the fixed
evidence-message vector collecting all
$\overline{\mu}_{\mathcal{E}_k}^{i,j,\tau}$
($i\in[\mathfrak{X}_k]$,
$j\in\{0\}\cup[\mathfrak{Y}_k]$,
$\tau\in[\mathfrak{M}_k]$) and
$\overline{\mu}_{\mathcal{E}_k}^{0,j}$
($j\in[\mathfrak{Y}_k]$), and let
$\boldsymbol{\mu}_{\mathcal{T}_k}
=(\overline{\mu}_{\mathcal{T}_k}^{i,j,\tau})
_{i=1,j=1,\tau=1}^
{\mathfrak{X}_k,\mathfrak{Y}_k,\mathfrak{M}_k}$
and
$\boldsymbol{\mu}_{\mathcal{M}_k}
=(\overline{\mu}_{\mathcal{M}_k}^{i,j,\tau})
_{i=1,j=1,\tau=1}^
{\mathfrak{X}_k,\mathfrak{Y}_k,\mathfrak{M}_k}$
denote the iterative constraint-message vectors.
Note that
$\dim(\boldsymbol{\mu}_{\mathcal{T}_k})
=\dim(\boldsymbol{\mu}_{\mathcal{M}_k})
=\mathfrak{X}_k\mathfrak{Y}_k\mathfrak{M}_k$,
which is smaller than
$\dim(\boldsymbol{\mu}_{\mathcal{E}_k})$,
reflecting that the missed-detection evidence
$\overline{\mu}_{\mathcal{E}_k}^{i,0,\tau}$ and the
clutter evidence
$\overline{\mu}_{\mathcal{E}_k}^{0,j}$
enter~\eqref{Eq:muTk}--\eqref{Eq:muMk} as fixed
scalars rather than iterative unknowns.
Let
$\boldsymbol{\mu}\in(0,+\infty)
^{\mathfrak{X}_k\mathfrak{Y}_k\mathfrak{M}_k}$
denote a generic message vector with elements
$\overline{\mu}_{i,j,\tau}$ for
$i\in[\mathfrak{X}_k]$, $j\in[\mathfrak{Y}_k]$,
$\tau\in[\mathfrak{M}_k]$, representing either
$\boldsymbol{\mu}_{\mathcal{T}_k}$ or
$\boldsymbol{\mu}_{\mathcal{M}_k}$.
Let
$\boldsymbol{\mu}^{(n)}\triangleq
\bigl((\boldsymbol{\mu}_{\mathcal{T}_k}^{(n)})^\top,\,
(\boldsymbol{\mu}_{\mathcal{M}_k}^{(n)})^\top
\bigr)^\top$
denote the combined message vector at the $n$-th
iteration. The overall BP procedure is summarized in
Algorithm~\ref{alg1}.

\begin{algorithm}[H]
\caption{The BP algorithm for MPDA.}\label{alg1}
\begin{algorithmic}[1]
\REQUIRE $\mathfrak{X}_k$, $\mathfrak{Y}_k$,
  $\mathfrak{M}_k$, fixed evidence messages
  $\boldsymbol{\mu}_{\mathcal{E}_k}$,
  convergence criterion~$\delta$.
\ENSURE Beliefs
  $b_A(a_k^{i,j,\tau}=1)$.
\STATE Initialize all elements of
  $\boldsymbol{\mu}^{(0)}$ in $(0,+\infty)$,
  set $n=1$ and $e>\delta$;
\WHILE{$e>\delta$}
  \FOR{all $i\in[\mathfrak{X}_k]$,
    $j\in[\mathfrak{Y}_k]$,
    $\tau\in[\mathfrak{M}_k]$}
    \STATE Update $\boldsymbol{\mu}^{(n)}$
      via~\eqref{Eq:muTk}--\eqref{Eq:muMk}
      using $\boldsymbol{\mu}^{(n-1)}$;
  \ENDFOR
  \STATE $e\leftarrow
    \|\boldsymbol{\mu}^{(n)}
    -\boldsymbol{\mu}^{(n-1)}\|_\infty$;
  \STATE $n\leftarrow n+1$;
\ENDWHILE
\FOR{all $i\in[\mathfrak{X}_k]$,
  $j\in[\mathfrak{Y}_k]$,
  $\tau\in[\mathfrak{M}_k]$}
  \STATE Compute $b_A(a_k^{i,j,\tau}=1)$
    via~\eqref{Eq:belief};
\ENDFOR
\RETURN $b_A(a_k^{i,j,\tau}=1)$.
\end{algorithmic}
\end{algorithm}

Algorithm~\ref{alg1} restates the inner BP procedure of Module~3 
in the JDT-VB framework~\cite{Lan2020Joint}.
In the present work, 
we consider each execution of this inner BP loop, in which 
$\boldsymbol{\mu}_{\mathcal{E}_k}$ is treated as a fixed strictly 
positive vector and only $\boldsymbol{\mu}_{\mathcal{T}_k}$ and 
$\boldsymbol{\mu}_{\mathcal{M}_k}$ are iteratively updated. Therefore, 
the convergence analysis in Section~\ref{section:multipath_convergence} 
concerns this inner BP execution, independently of the outer JDT-VB 
updates or any other tracking or detection framework used to obtain 
$\boldsymbol{\mu}_{\mathcal{E}_k}$.

The computational cost of Algorithm~\ref{alg1} is
dominated by the iterative message updates. At each
iteration, the messages
$\overline{\mu}_{\mathcal{T}_k}^{i,j,\tau}$ and
$\overline{\mu}_{\mathcal{M}_k}^{i,j,\tau}$ must be
computed for each of the
$\mathfrak{X}_k\mathfrak{Y}_k\mathfrak{M}_k$
non-null variable nodes. Letting $r_{\mathrm{lbp}}$
denote the number of iterations, the overall
complexity is
$C_{\mathrm{LBP}}
=\mathcal{O}(r_{\mathrm{lbp}}\,
\mathfrak{X}_k\mathfrak{Y}_k\mathfrak{M}_k)$.

\section{Convergence Analysis of BP for MPDA}
\label{section:multipath_convergence}

Although Fig.~\ref{dataAssociation_Multipath2} contains loops, 
we prove that Algorithm~\ref{alg1} converges to a unique fixed point.

\subsection{Structural Relation to Two-Way DA}
\label{sec:standardDA}

Two-way DA refers to the one-to-one correspondence between
targets and measurements, where each target generates at most
one measurement per scan and each measurement originates from
at most one target~\cite{Williams2014Approximate}.
MPDA extends this to a three-way correspondence among
targets, measurements, and propagation paths.
In the MPDA considered here, a target may generate multiple
measurements through distinct propagation paths, while at most
one measurement is generated through each propagation path of
each target.
Therefore, MPDA reduces to two-way DA only when
$\mathfrak{M}_k=1$.

In each execution of BP for MPDA, the evidence
messages~$\boldsymbol{\mu}_{\mathcal{E}_k}$ are fixed,
whereas only the constraint messages
$\boldsymbol{\mu}_{\mathcal{T}_k}$ and
$\boldsymbol{\mu}_{\mathcal{M}_k}$ are updated
iteratively
through~\eqref{Eq:muTk}--\eqref{Eq:muMk}.
Remark~1 of~\cite{Lan2020Joint} observed that the
convergence of BP in MPDA can be related to the
two-way DA analysis
in~\cite{Williams2014Approximate} by treating each
(target, path) pair as a pseudo-target.
However, Remark in~\cite{Lan2020Joint} was not
accompanied by an explicit convergence theorem or a
complete proof for the MPDA message
updates~\eqref{Eq:muTk}--\eqref{Eq:muMk}.
By reexamining this pseudo-target argument more closely,
while the contraction property of each message update
can be verified on compact subsets
via~\cite[Lemma~1 and Lemma~2]{Williams2014Approximate},
the explicit construction of a positively invariant
compact subset from which the Banach fixed-point theorem
can be applied is not straightforward and was not
provided in~\cite{Lan2020Joint}.
The purpose of the analysis below is to provide such a
proof for Algorithm~\ref{alg1}.

Specifically, based on the structural relation described
above, Proposition~\ref{prop:contraction} applies the
contraction results
of~\cite[Lemma~1 and Lemma~2]{Williams2014Approximate}
to show that the MPDA message updates~$g(\cdot)$
and~$h(\cdot)$ are strict contractions on compact subsets of~$(0,+\infty)^{\mathfrak{X}_k\mathfrak{Y}_k\mathfrak{M}_k}$, with contraction factors that
depend on the subset.
However, the contraction property on a compact subset
does not by itself imply convergence from an arbitrary
strictly positive initialization, since one must
additionally identify a positively invariant
compact subset that (i)~is entered by the message
sequence after finitely many iterations and (ii)~on
which the contraction holds.
For two-way DA, \cite{Williams2014Approximate} noted that the message iterates are contained in a compact subset, but did not provide its explicit construction. For the original MPDA message updates in~\eqref{Eq:muTk}--\eqref{Eq:muMk}, the corresponding positively invariant compact subset has not been established in \cite{Lan2020Joint}.
Theorem~\ref{Theorem1} below explicitly constructs this
subset with concrete bounds and applies the Banach
fixed-point theorem to establish convergence of BP for
MPDA.

We next define the metric space and derive the contraction
conditions required for Theorem~\ref{Theorem1}.

\subsection{Metric Space Formulation}

For notational simplicity, the time index~$k$ is
omitted in what follows. Recall that the evidence-message
vector~$\boldsymbol{\mu}_{\mathcal{E}}$ serves as a fixed
input. Therefore, the message
$\overline{\mu}_{\mathcal{T}}^{i,j,\tau}$ depends only on
the previous $\overline{\mu}_{\mathcal{M}}^{i,j,\tau}$ and
vice versa in~\eqref{Eq:muTk}--\eqref{Eq:muMk}. Let
$\boldsymbol{\mu}^{\mathcal{T}}
=\boldsymbol{g}(\boldsymbol{\mu}^{\mathcal{M}})$ and
$\boldsymbol{\mu}^{\mathcal{M}}
=\boldsymbol{h}(\boldsymbol{\mu}^{\mathcal{T}})$
represent the message updates
in~\eqref{Eq:muTk} and~\eqref{Eq:muMk} in vector form.
The domains of both $\boldsymbol{g}(\cdot)$ and
$\boldsymbol{h}(\cdot)$ are $(0,\infty)^t$, where
$t=\mathfrak{X}_k\mathfrak{Y}_k\mathfrak{M}_k$.
Accordingly, their ranges are also $(0,\infty)^t$.

Consider a metric space
$\mathcal{X}=(0,+\infty)^t$ equipped with a distance
metric $d:\mathcal{X}\times\mathcal{X}\to[0,+\infty)$.
A function $f:\mathcal{X}\to\mathcal{X}$ is called a
contraction mapping if there exists
$\alpha\in[0,1)$ such that
$d(f(x),f(y))\le\alpha\,d(x,y)$ for all
$x,y\in\mathcal{X}$. Moreover, if $\mathcal{X}$ is
complete, then any sequence resulting from repeated
application of~$f$ converges to a unique fixed
point~\cite{keeler1969theorem}.

While Algorithm~\ref{alg1} uses the $L_\infty$ norm as
a stopping criterion, establishing the contraction
property requires an appropriate metric. Inspired
by~\cite{Williams2014Approximate}, we define the
logarithmic distance,
\begin{equation}\label{Eq:logdist}
d(\boldsymbol{\mu},\tilde{\boldsymbol{\mu}})
= \max_{i,j,\tau}
  \Bigl|\log\frac{\overline{\mu}_{i,j,\tau}}
  {\tilde{\overline{\mu}}_{i,j,\tau}}\Bigr|.
\end{equation}
One can verify that $d(\cdot,\cdot)$ is a valid distance
metric. As shown in Theorem~\ref{Theorem1}, the
messages enter and remain in a compact subset of
$(0,+\infty)^t$ during iteration. On such a subset,
$d(\boldsymbol{\mu},\tilde{\boldsymbol{\mu}})\to 0$ if
and only if
$\|\boldsymbol{\mu}-\tilde{\boldsymbol{\mu}}\|_\infty
\to 0$, ensuring consistency with the stopping
criterion in Algorithm~\ref{alg1}.

\subsection{Contraction Properties of Message Updates}

We show that $\boldsymbol{g}(\cdot)$ and
$\boldsymbol{h}(\cdot)$
in~\eqref{Eq:muTk}--\eqref{Eq:muMk} are contraction
mappings by algebraically reformulating them into the
canonical fractional form of~\cite{Williams2014Approximate},
enabling direct application
of~\cite[Lemma~1 and Lemma~2]{Williams2014Approximate}.

\begin{proposition}\label{prop:contraction}
Consider the distance metric $d(\cdot,\cdot)$ defined
in~\eqref{Eq:logdist}, and let
$\alpha(L,c)=\frac{\log\frac{1+cL}{1+c}}{\log L}$
denote the contraction factor
from~\cite[Lemma~1]{Williams2014Approximate}, defined
for $L>1$ and $c>0$, with
$\alpha(L,c)\in(0,1)$ and monotonically increasing
in~$L$. For any compact subsets
$\Omega_{\mathcal{M}}=[L_{\mathcal{M}},U_{\mathcal{M}}]^t
\subset\mathcal{X}$ and
$\Omega_{\mathcal{T}}=[L_{\mathcal{T}},U_{\mathcal{T}}]^t
\subset\mathcal{X}$, define
$\bar{L}_{\mathcal{M}}\triangleq
U_{\mathcal{M}}/L_{\mathcal{M}}>1$ and
$\bar{L}_{\mathcal{T}}\triangleq
U_{\mathcal{T}}/L_{\mathcal{T}}>1$. Then, for all
$\boldsymbol{\mu}^{\mathcal{M}},
\tilde{\boldsymbol{\mu}}^{\mathcal{M}}
\in\Omega_{\mathcal{M}}$ and
$\boldsymbol{\mu}^{\mathcal{T}},
\tilde{\boldsymbol{\mu}}^{\mathcal{T}}
\in\Omega_{\mathcal{T}}$,
\begin{align}
d\bigl(\boldsymbol{g}(\boldsymbol{\mu}^{\mathcal{M}}),
  \boldsymbol{g}(\tilde{\boldsymbol{\mu}}^{\mathcal{M}})
  \bigr)
&\le\alpha(\bar{L}_{\mathcal{M}},C_{\mathcal{M}}^*)\,
  d(\boldsymbol{\mu}^{\mathcal{M}},
  \tilde{\boldsymbol{\mu}}^{\mathcal{M}}),
  \label{Eq:g_contraction}\\
d\bigl(\boldsymbol{h}(\boldsymbol{\mu}^{\mathcal{T}}),
  \boldsymbol{h}(\tilde{\boldsymbol{\mu}}^{\mathcal{T}})
  \bigr)
&\le\alpha(\bar{L}_{\mathcal{T}},C_{\mathcal{T}}^*)\,
  d(\boldsymbol{\mu}^{\mathcal{T}},
  \tilde{\boldsymbol{\mu}}^{\mathcal{T}}),
  \label{Eq:h_contraction}
\end{align}
where
\begin{align}
C_{\mathcal{M}}^*
&\triangleq\max_{i,j,\tau}\max_{\boldsymbol{\mu}^{\mathcal{M}}
  \in\Omega_{\mathcal{M}}}
  c_{i,j,\tau}^{\mathcal{M}}(\boldsymbol{\mu}^{\mathcal{M}}),
  \label{Eq:CM}\\
C_{\mathcal{T}}^*
&\triangleq\max_{i,j,\tau}\max_{\boldsymbol{\mu}^{\mathcal{T}}
  \in\Omega_{\mathcal{T}}}
  c_{i,j,\tau}^{\mathcal{T}}(\boldsymbol{\mu}^{\mathcal{T}}),
  \label{Eq:CT}
\end{align}
with
$c_{i,j,\tau}^{\mathcal{M}}(\boldsymbol{\mu}^{\mathcal{M}})
=\frac{1}{\overline{\mu}_{\mathcal{E}}^{i,0,\tau}}
\sum_{\substack{j_1=1\\j_1\neq j}}^{\mathfrak{Y}_k}
\overline{\mu}_{\mathcal{E}}^{i,j_1,\tau}\,
\overline{\mu}_{\mathcal{M}}^{i,j_1,\tau}$
and
$c_{i,j,\tau}^{\mathcal{T}}(\boldsymbol{\mu}^{\mathcal{T}})
=\frac{1}{\overline{\mu}_{\mathcal{E}}^{0,j}}
\sum_{\substack{i_1=1,\tau_1=1\\ (i_1,\tau_1)\neq(i,\tau)}}^{\mathfrak{X}_k,\mathfrak{M}_k}
\overline{\mu}_{\mathcal{E}}^{i_1,j,\tau_1}\,
\overline{\mu}_{\mathcal{T}}^{i_1,j,\tau_1}$.
\end{proposition}

\begin{proof}
Dividing numerator and denominator
of~\eqref{Eq:muTk}--\eqref{Eq:muMk} by the strictly
positive constants
$\overline{\mu}_{\mathcal{E}}^{i,0,\tau}$ and
$\overline{\mu}_{\mathcal{E}}^{0,j}$ yields
\begin{align*}
g_{i,j,\tau}(\boldsymbol{\mu}^{\mathcal{M}})
&=\frac{1/\overline{\mu}_{\mathcal{E}}^{i,0,\tau}}
  {1+c_{i,j,\tau}^{\mathcal{M}}
  (\boldsymbol{\mu}^{\mathcal{M}})},& \!\!\!
h_{i,j,\tau}(\boldsymbol{\mu}^{\mathcal{T}})
&=\frac{1/\overline{\mu}_{\mathcal{E}}^{0,j}}
  {1+c_{i,j,\tau}^{\mathcal{T}}
  (\boldsymbol{\mu}^{\mathcal{T}})},
\end{align*}
which are of the canonical fractional form of~\cite[(20)]{Williams2014Approximate}.

Moreover, for any $\boldsymbol{\mu}^{\mathcal{M}},\tilde{\boldsymbol{\mu}}^{\mathcal{M}}\in\Omega_{\mathcal{M}}$, we have $d(\boldsymbol{\mu}^{\mathcal{M}},\tilde{\boldsymbol{\mu}}^{\mathcal{M}})\le \log(U_{\mathcal{M}}/L_{\mathcal{M}})=\log\bar{L}_{\mathcal{M}}$.
Similarly, for any $\boldsymbol{\mu}^{\mathcal{T}},\tilde{\boldsymbol{\mu}}^{\mathcal{T}}\in\Omega_{\mathcal{T}}$, $d(\boldsymbol{\mu}^{\mathcal{T}},\tilde{\boldsymbol{\mu}}^{\mathcal{T}})\le \log(U_{\mathcal{T}}/L_{\mathcal{T}})=\log\bar{L}_{\mathcal{T}}$. 
Since $c_{i,j,\tau}^{\mathcal{M}}(\cdot)$ and $c_{i,j,\tau}^{\mathcal{T}}(\cdot)$ are continuous, 
$\Omega_{\mathcal{M}}$ and $\Omega_{\mathcal{T}}$ are compact, the maxima in \eqref{Eq:CM}--\eqref{Eq:CT} are attained and finite.
For fixed $L>1$, $\alpha(L,c)$ is strictly increasing in $c$, since $\frac{\partial}{\partial c}\log\frac{1+cL}{1+c}=\frac{L-1}{(1+cL)(1+c)}>0$. Applying \cite[Lemma~1 and Lemma~2]{Williams2014Approximate}, together with this monotonicity in $c$, yields \eqref{Eq:g_contraction} and \eqref{Eq:h_contraction}, where both $\alpha(\bar{L}_{\mathcal{M}},C_{\mathcal{M}}^*)$ and $\alpha(\bar{L}_{\mathcal{T}},C_{\mathcal{T}}^*)$ are strictly less than one, completing the proof.

\end{proof}

\subsection{Convergence Theorem}

Building upon Proposition~\ref{prop:contraction}, we now establish the 
convergence of the loopy BP updates in Algorithm~\ref{alg1}.

\begin{theorem}\label{Theorem1}
Consider the message space
$\mathcal{X}=(0,+\infty)^t$ with the logarithmic
distance~\eqref{Eq:logdist}. Define the update mapping
\begin{equation}\label{Eq:F}
\boldsymbol{F}(\boldsymbol{\mu}^{\mathcal{M}},
\boldsymbol{\mu}^{\mathcal{T}})
\triangleq\bigl(\boldsymbol{h}(\boldsymbol{\mu}^{\mathcal{T}}),\,
\boldsymbol{g}(\boldsymbol{\mu}^{\mathcal{M}})\bigr).
\end{equation}
Then, for any initialization
$(\boldsymbol{\mu}^{\mathcal{M},(0)},
\boldsymbol{\mu}^{\mathcal{T},(0)})
\in\mathcal{X}\times\mathcal{X}$, the BP updates
induced by~$\boldsymbol{F}$ converge to a unique fixed
point.
\end{theorem}

\begin{proof}
We first consider the degenerate cases. If
$\mathfrak{Y}_k=1$, the sum term
in~\eqref{Eq:muTk} is empty for every $(i,\tau)$, so
$\overline{\mu}_{\mathcal{T}}^{i,1,\tau}
=1/\overline{\mu}_{\mathcal{E}}^{i,0,\tau}$ is
constant; substituting it into~\eqref{Eq:muMk} then
yields a constant
$\overline{\mu}_{\mathcal{M}}^{i,1,\tau}$, and
convergence is immediate. The case
$\mathfrak{X}_k\mathfrak{M}_k=1$ follows by a
symmetric argument applied to~\eqref{Eq:muMk}.
It remains to consider the nondegenerate case $\mathfrak{Y}_k\geq 2,
\mathfrak{X}_k\mathfrak{M}_k\geq 2$.

We first show that $(\mathcal{X},d)$ is complete. Define the 
elementwise logarithmic mapping $\varphi(\boldsymbol{\mu})=\log\boldsymbol{\mu}
\in\mathbb{R}^t$. Then, by \eqref{Eq:logdist}, 
$d(\boldsymbol{\mu},\tilde{\boldsymbol{\mu}})
=\|\varphi(\boldsymbol{\mu})-\varphi(\tilde{\boldsymbol{\mu}})\|_\infty$.
Since $(\mathbb{R}^t,\|\cdot\|_\infty)$ is complete and $\varphi$ is 
bijective with continuous inverse $\varphi^{-1}(\mathbf{x})=\exp(\mathbf{x})$,
$(\mathcal{X},d)$ is complete. 
Consequently, the product space 
$(\mathcal{X}\times\mathcal{X}, d_{\max})$, with $d_{\max}\bigl((\bm{\mu}^{\mathcal{M}},\bm{\mu}^{\mathcal{T}}),
(\tilde{\bm{\mu}}^{\mathcal{M}},\tilde{\bm{\mu}}^{\mathcal{T}})\bigr)
\triangleq
\max\bigl\{d(\bm{\mu}^{\mathcal{M}},\tilde{\bm{\mu}}^{\mathcal{M}}),\,     d(\bm{\mu}^{\mathcal{T}},\tilde{\bm{\mu}}^{\mathcal{T}})\bigr\}$
is also complete.

Define the following constants from the fixed evidence
messages~$\boldsymbol{\mu}_{\mathcal{E}}$,
\begin{align*}
\kappa_{m,\min}
  &\triangleq\min_{i,\tau}
    \overline{\mu}_{\mathcal{E}}^{i,0,\tau}>0,&
\kappa_{m,\max}
  &\triangleq\max_{i,\tau}
    \overline{\mu}_{\mathcal{E}}^{i,0,\tau},\\
\kappa_{c,\min}
  &\triangleq\min_{j}
    \overline{\mu}_{\mathcal{E}}^{0,j}>0,&
\kappa_{c,\max}
  &\triangleq\max_{j}
    \overline{\mu}_{\mathcal{E}}^{0,j},\\
\kappa_{e,\max}
  &\triangleq\max_{i>0,j>0,\tau}
    \overline{\mu}_{\mathcal{E}}^{i,j,\tau}.
\end{align*}

Next, we show that for any initial
$(\boldsymbol{\mu}^{\mathcal{M},(0)},
\boldsymbol{\mu}^{\mathcal{T},(0)})
\in\mathcal{X}\times\mathcal{X}$,
the messages after the second iteration enter a
positively invariant compact subset
$\hat{\Omega}\subset\mathcal{X}\times\mathcal{X}$.

From~\eqref{Eq:muTk}--\eqref{Eq:muMk}, the bounds
$0<\overline{\mu}_{\mathcal{T}}^{i,j,\tau}
\le 1/\kappa_{m,\min}
\triangleq U_{\mathcal{T}}$
and
$0<\overline{\mu}_{\mathcal{M}}^{i,j,\tau}
\le 1/\kappa_{c,\min}
\triangleq U_{\mathcal{M}}$
ensure that, after the first synchronous update,
$\boldsymbol{\mu}^{\mathcal{T},(1)}
\in(0,U_{\mathcal{T}}]^t$ and
$\boldsymbol{\mu}^{\mathcal{M},(1)}
\in(0,U_{\mathcal{M}}]^t$.

At the second iteration, substituting
$\boldsymbol{\mu}^{\mathcal{M},(1)}
\in(0,U_{\mathcal{M}}]^t$ into~\eqref{Eq:muTk} yields
$c_{i,j,\tau}^{\mathcal{M}}
(\boldsymbol{\mu}^{\mathcal{M},(1)})
\le(\mathfrak{Y}_k-1)\,
\kappa_{e,\max}\,U_{\mathcal{M}}
/\kappa_{m,\min}
\triangleq\bar{\kappa}^{\mathcal{M}}$,
which implies
$\overline{\mu}_{\mathcal{T}}^{i,j,\tau}
\ge\frac{1}{\kappa_{m,\max}
(1+\bar{\kappa}^{\mathcal{M}})}
\triangleq L_{\mathcal{T}}>0$.
Hence,
$\boldsymbol{\mu}^{\mathcal{T},(2)}
\in\Omega_{\mathcal{T}}
\triangleq[L_{\mathcal{T}},U_{\mathcal{T}}]^t$.
Analogously, substituting
$\boldsymbol{\mu}^{\mathcal{T},(1)}
\in(0,U_{\mathcal{T}}]^t$ into~\eqref{Eq:muMk} yields
$c_{i,j,\tau}^{\mathcal{T}}
(\boldsymbol{\mu}^{\mathcal{T},(1)})
\le(\mathfrak{X}_k\mathfrak{M}_k-1)\,
\kappa_{e,\max}\,U_{\mathcal{T}}
/\kappa_{c,\min}
\triangleq\bar{\kappa}^{\mathcal{T}}$,
which implies
$\overline{\mu}_{\mathcal{M}}^{i,j,\tau}
\ge\frac{1}{\kappa_{c,\max}
(1+\bar{\kappa}^{\mathcal{T}})}
\triangleq L_{\mathcal{M}}>0$.
Hence,
$\boldsymbol{\mu}^{\mathcal{M},(2)}
\in\Omega_{\mathcal{M}}
\triangleq[L_{\mathcal{M}},U_{\mathcal{M}}]^t$.

Since $\Omega_{\mathcal M}=[L_{\mathcal M},U_{\mathcal M}]^t$ and $\Omega_{\mathcal T}=[L_{\mathcal T},U_{\mathcal T}]^t$, for any $\boldsymbol{\mu}^{\mathcal M}\in\Omega_{\mathcal M}$ and $\boldsymbol{\mu}^{\mathcal T}\in\Omega_{\mathcal T}$, we have
$c_{i,j,\tau}^{\mathcal M}(\boldsymbol{\mu}^{\mathcal M})\le \bar{\kappa}^{\mathcal M}$ and
$c_{i,j,\tau}^{\mathcal T}(\boldsymbol{\mu}^{\mathcal T})\le \bar{\kappa}^{\mathcal T}$. 
Hence, $g_{i,j,\tau}(\boldsymbol{\mu}^{\mathcal{M}})\in[L_{\mathcal{T}},U_{\mathcal{T}}]$ 
and $h_{i,j,\tau}(\boldsymbol{\mu}^{\mathcal{T}})\in[L_{\mathcal{M}},U_{\mathcal{M}}]$. 
Therefore, $\boldsymbol{g}(\Omega_{\mathcal{M}})\subseteq\Omega_{\mathcal{T}}$, 
$\boldsymbol{h}(\Omega_{\mathcal{T}})\subseteq\Omega_{\mathcal{M}}$.
Let $\hat{\Omega} \triangleq \Omega_{\mathcal{M}} \times \Omega_{\mathcal{T}}$, then we have $\boldsymbol{F}(\hat{\Omega}) \subseteq \hat{\Omega}$; that is, $\hat{\Omega}$ is positively invariant.
Since $(\boldsymbol{\mu}^{\mathcal{M},(2)},
\boldsymbol{\mu}^{\mathcal{T},(2)})\in\hat{\Omega}$, it follows that 
$(\boldsymbol{\mu}^{\mathcal{M},(n)},\boldsymbol{\mu}^{\mathcal{T},(n)})\in\hat{\Omega}$ 
for all $n\ge 2$.

By Proposition~\ref{prop:contraction}, where
$\bar{L}_{\mathcal{M}}\triangleq U_{\mathcal{M}}/L_{\mathcal{M}}>1$ and
$\bar{L}_{\mathcal{T}}\triangleq U_{\mathcal{T}}/L_{\mathcal{T}}>1$,
$\boldsymbol{g}$ and $\boldsymbol{h}$ satisfy \eqref{Eq:g_contraction} and
\eqref{Eq:h_contraction} on $\Omega_{\mathcal{M}}$ and $\Omega_{\mathcal{T}}$,
respectively. Therefore, for any
$(\boldsymbol{\mu}^{\mathcal{M}},\boldsymbol{\mu}^{\mathcal{T}}),
(\tilde{\boldsymbol{\mu}}^{\mathcal{M}},\tilde{\boldsymbol{\mu}}^{\mathcal{T}})
\in\hat{\Omega}$,
\begin{align}
&d_{\max}\bigl(\boldsymbol{F}(\boldsymbol{\mu}^{\mathcal{M}},\boldsymbol{\mu}^{\mathcal{T}}),
\boldsymbol{F}(\tilde{\boldsymbol{\mu}}^{\mathcal{M}},
\tilde{\boldsymbol{\mu}}^{\mathcal{T}})\bigr) \nonumber\\
&=\max\bigl\{d(\boldsymbol{h}(\boldsymbol{\mu}^{\mathcal{T}}),
\boldsymbol{h}(\tilde{\boldsymbol{\mu}}^{\mathcal{T}})),\,
d(\boldsymbol{g}(\boldsymbol{\mu}^{\mathcal{M}}),
\boldsymbol{g}(\tilde{\boldsymbol{\mu}}^{\mathcal{M}}))\bigr\} \nonumber\\
&\le\max\bigl\{\alpha(\bar{L}_{\mathcal{T}},C_{\mathcal{T}}^*)\,
d(\boldsymbol{\mu}^{\mathcal{T}},\tilde{\boldsymbol{\mu}}^{\mathcal{T}}),\,
\alpha(\bar{L}_{\mathcal{M}},C_{\mathcal{M}}^*)\,
d(\boldsymbol{\mu}^{\mathcal{M}},\tilde{\boldsymbol{\mu}}^{\mathcal{M}})\bigr\} \nonumber\\
&\le\alpha_{\mathrm{sync}}\,d_{\max}\bigl((\boldsymbol{\mu}^{\mathcal{M}},
\boldsymbol{\mu}^{\mathcal{T}}),(\tilde{\boldsymbol{\mu}}^{\mathcal{M}},
\tilde{\boldsymbol{\mu}}^{\mathcal{T}})\bigr),
\end{align}
where $\alpha_{\mathrm{sync}}\triangleq
\max\{\alpha(\bar{L}_{\mathcal{M}},C_{\mathcal{M}}^*),
\alpha(\bar{L}_{\mathcal{T}},C_{\mathcal{T}}^*)\}<1$.
Hence, $\boldsymbol{F}$ is a strict contraction on $\hat{\Omega}$.

Since $\hat{\Omega}$ is a compact subset of the complete
space $(\mathcal{X}\times\mathcal{X},d_{\max})$, it is
itself complete.
The Banach fixed-point theorem~\cite{Banach1922} therefore guarantees a unique
fixed point $(\boldsymbol{\mu}^{\mathcal{M},*},\boldsymbol{\mu}^{\mathcal{T},*})\in\hat{\Omega}$,
to which the iterates converge for any initialization $(\boldsymbol{\mu}^{\mathcal{M},(0)},\boldsymbol{\mu}^{\mathcal{T},(0)})
\in\mathcal{X}\times\mathcal{X}$.
The proof is complete.

\end{proof}

\textbf{Remark 1:}
Theorem~\ref{Theorem1} is also consistent with
existing theory in the single-path case
($\mathfrak{M}_k=1$), where the MPDA constraints
reduce to those of two-way DA and the marginal
posteriors in~\eqref{Eq:belief} coincide with those
in~\cite{Williams2014Approximate}. Specifically, by
identifying
$\psi_i(j)=
\mu_{\mathcal{E}}^{i,j,1}/
(\mu_{\mathcal{E}}^{i,0,1}\mu_{\mathcal{E}}^{0,j})$
and
$\psi_i(0)=1$,
BP yields identical approximate marginals. 

\section{Numerical Experiments}\label{Sec:Experiments}

While Section~\ref{section:multipath_convergence} has
established the theoretical convergence of BP for MPDA,
we now provide an empirical evaluation in an OTHR target
tracking scenario, where a single target may generate
multiple measurements through different ionospheric
propagation paths. 
Consistent with the 
convergence analysis, we restrict attention to the inner BP
procedure for MPDA inference under fixed evidence messages;
the outer VB loop of JDT-VB is not implemented in the
experiments.

\subsection{Target and Measurement Model}\label{Sec:model}

We consider an MTT scenario for OTHR, where targets follow a nearly constant velocity dynamic model formulated in ground coordinates~\cite{Pulford1998}.
The radar receiver is located at the origin, and the transmitter is placed at a distance $d_{\mathrm{tx}}$ along the \(x\)-axis. 
The target kinematic state at scan $k$ is denoted as ${x}_k = \begin{bmatrix} g_k & \dot{g}_k & \vartheta_k & \dot{\vartheta}_k \end{bmatrix}^\top$, where $g_k$ and $\dot{g}_k$ denote the ground range and ground range rate, while $\vartheta_k$ and $\dot{\vartheta}_k$ represent the bearing angle and its rate, respectively.

To account for multipath propagation, we adopt the well-established ionospheric reflection model detailed in~\cite{Pulford1998}. 
Assuming two dominant ionospheric layers (E and F), the transmitted signals yield four possible propagation modes, denoted by $\tau \in \{\text{E-E, E-F, F-E, F-F}\}$. At each scan, the OTHR receives the slant measurement vector $y_k = [r_k, \dot{r}_k, \zeta_k]^\top$, comprising the slant range, slant range-rate, and azimuth. 
The nonlinear coordinate transformations, along with the corresponding state transition and observation noise covariance matrices, follow~\cite{Pulford1998}.

\smallskip\noindent\textbf{Scenario parameters.}\;
The surveillance region is assumed to be \([1500,\,2000]\) km in range, and \([0.626,\,0.899]\) rad in azimuth.
The slant range-rate is assumed to be in $[-0.12,\,0.12]$ km/s.
Measurement errors for all propagation paths are modeled as zero-mean Gaussian with standard deviations \(\sigma_r = 5\)~km (slant range), \(\sigma_{\dot{r}} = 0.001\)~km/s (slant range rate), and \(\sigma_\zeta = 0.003\)~rad (azimuth).

Without loss of generality, consider that 
\(p_{\text{d}}^{\tau} = p_{\text{d}},\ \tau = 1,2,3,4\).
Clutter is modeled as a Poisson point process 
over the measurement space, with surveillance volume 
\(V_k = 500\ \mathrm{km} \times 0.24\ \mathrm{km/s} \times 0.273\ \mathrm{rad} \).
The number of clutter measurements \(N_{c,k}\) at each scan 
follows a Poisson distribution with mean 
\(\lambda_{c,k}\triangleq \mathbb{E}[N_{c,k}]=\lambda V_k\), i.e., 
\(N_{c,k} \sim \mathrm{Poisson}(\lambda_{c,k})\).
Equivalently, when \(\lambda_{c,k}\) is prescribed in the simulations, the corresponding clutter spatial density is \(\lambda=\lambda_{c,k}/V_k\).

We set $d_{\mathrm{tx}} = 100$~km, ionospheric layer heights $H_{\text{E}} = 100$~km (E-layer), and $H_{\text{F}} = 260$~km (F-layer).
In all experiments, the convergence
threshold is fixed at $\delta=10^{-5}$.

\subsection{Experimental Setup of Experiments}

We design four distinct experiments to evaluate the
marginal approximation accuracy, convergence behavior,
and tracking accuracy of BP.
Unless otherwise specified, the simulation comprises 100 time steps with a sampling interval of $T = 10$~s, and the targets are initially uniformly spaced on a circle of radius $\rho = 50$~km and move toward the center at a constant speed of $0.1$~km/s.
The targets intersect at the origin between time steps 40 and 60 before moving outward, as depicted in Fig.~\ref{fig:3TargetsPolar}.

\begin{figure}[!htp]
\centering
\includegraphics[width=0.3\textwidth]{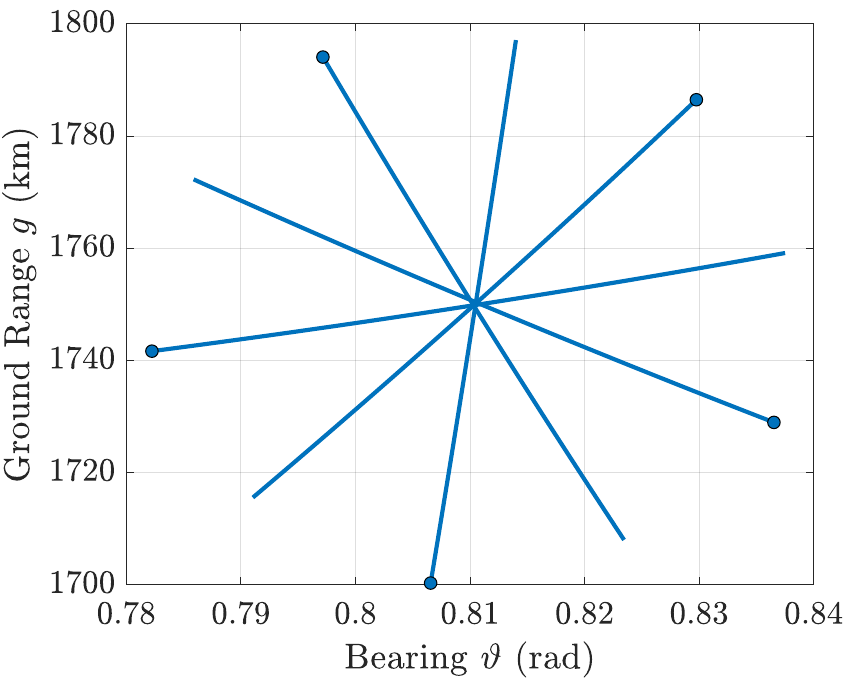}
\caption{True trajectories of five targets in the range-bearing plane over 100 time steps.}
\label{fig:3TargetsPolar}
\end{figure}

At each scan~$k$, an extended Kalman filter (EKF) based on the nonlinear OTHR measurement model in~\cite{Pulford1998} is
employed for target kinematic state estimation. The EKF
time-prediction step provides, for target~$i$, the predicted
state $\widehat{x}_{i,k}^{-}$ and covariance~$P_{i,k}^{-}$.
For each target--path pair~$(i,\tau)$, the predicted
measurement is
$\widehat{y}_{i,\tau,k}^{-}
=h_{\tau,k}(\widehat{x}_{i,k}^{-})$
with innovation covariance
$S_{i,\tau,k}
=H_{i,\tau,k}P_{i,k}^{-}H_{i,\tau,k}^{\mathsf{T}}
+R_{\tau,k}$,
where $H_{i,\tau,k}$ is the Jacobian of
$h_{\tau,k}(\cdot)$ at~$\widehat{x}_{i,k}^{-}$.
The corresponding predictive likelihood is
$\ell_{\tau,k}^{i,j}
\triangleq
\mathcal{N}(y_{j,k};
\widehat{y}_{i,\tau,k}^{-},S_{i,\tau,k})$.
Based on the predicted-measurement likelihood terms and the
detection and clutter models in the MPDA formulation
of~\cite{Pulford1998,Lan2020Joint}, the fixed evidence messages
supplied to Algorithm~\ref{alg1} are constructed as,
\begin{equation}\label{Eq:SimulationEvidenceMessages}
\overline{\mu}_{\mathcal{E}_k}^{i,j,\tau}
\triangleq
\frac{p_d^\tau \ell_{\tau,k}^{i,j}}{\lambda V_k},
\qquad
\overline{\mu}_{\mathcal{E}_k}^{i,0,\tau}
\triangleq
1-p_d^\tau,
\qquad
\overline{\mu}_{\mathcal{E}_k}^{0,j}
\triangleq
\frac{1}{V_k}.
\end{equation}
Algorithm~\ref{alg1} is then executed with
$\boldsymbol{\mu}_{\mathcal{E}_k}$ held fixed throughout
the inner BP iterations.
In Experiment~I, the converged 
BP beliefs are evaluated directly against exact marginal association probabilities 
in a single-scan inference problem. For the tracking experiments, the resulting 
association beliefs are used in the EKF-based state update.
No outer VB iteration or target existence-state update is
performed; hence, the experiments evaluate the inner BP
procedure under tracking-generated fixed evidence messages,
rather than the complete JDT-VB algorithm.
All results are averaged over $500$ Monte Carlo (MC) runs.

\subsubsection{Experiment~I}
To evaluate the accuracy of the converged BP beliefs,
we consider a single-scan MPDA inference problem with
$\mathfrak{X}_k=2$ targets and $\mathfrak{M}_k=2$ propagation
paths (F-E and F-F).
The scenario is kept small so that the exact marginal
association probabilities can be obtained by exhaustive
enumeration of all feasible MPDA events.
The two targets are uniformly spaced on a circle of radius
$\rho\in\{5,10,15,20,25\}$~km. Each target--path pair generates a measurement
independently with probability $p_{\mathrm{d}}\in\{0.6,0.9\}$, and clutter
measurements are drawn from a Poisson process with mean $\lambda_{c,k}=2$;
consequently, $\mathfrak{Y}_k$ varies across trials. The evidence messages
in~\eqref{Eq:SimulationEvidenceMessages} are constructed with
$\widehat{y}_{i,\tau,k}^{-}
=h_{\tau,k}(x_{i,k})$ and $S_{i,\tau,k}
=R_{\tau,k}$.
Following \cite{Williams2014Approximate}, the approximation
accuracy is evaluated by the average maximum error~(AME),
defined as the mean over MC trials of the largest
absolute difference between the BP belief and the exact marginal.

\subsubsection{Experiment II}
To examine the convergence of BP in a dense-target MTT scenario, we set the number of targets to $\mathfrak{X}_k = 100$ with four propagation paths. 
The detection probability is fixed at \(p_{\text{d}} = 0.9\), 
and the average number of clutter measurements per scan is set to 
\(\lambda_{c,k} = 90\).
We analyze the number of iterations required for the message 
residual $e$ in Algorithm~\ref{alg1} to satisfy the convergence 
criterion $e < \delta$.

\subsubsection{Experiment III}
To further evaluate the performance of BP relative to the MD-MHT method~\cite{MHT2013Sathyan}, we consider a scenario with $\mathfrak{X}_k = 3$ targets and two propagation paths (F-F and F-E). 
The simulations are conducted under various combinations of detection probability 
$p_{\text{d}} \in \{0.3, 0.6, 0.9\}$ and average number of clutter measurements
per scan $\lambda_{c,k} \in \{10, 20\}$.
Specifically, the MD-MHT implementation employs Murty's approximation within a track-oriented framework to enhance computational efficiency~\cite{Blackman2004MHT}.
We evaluate two variants of MD-MHT: one using single-scan association (denoted as MDMHT-1) and another using a two-scan sliding window association (denoted as MDMHT-2)~\cite{MHT2013Sathyan}.
Finally, we compare the tracking accuracy and average per-step execution time among the three methods.

\subsubsection{Experiment IV}
The final experiment investigates the impact of the number of targets and propagation paths on BP.
First, we fix the detection probability at $p_{\text{d}} = 0.9$, 
the average number of clutter measurements per scan at $\lambda_{c,k} = 20$, 
and the number of propagation paths to two (F-F and F-E), 
while varying the number of targets $\mathfrak{X}_k \in \{5, 10, 15, 20, 25, 30\}$.
Subsequently, we fix the number of targets at $\mathfrak{X}_k = 15$ and vary the number of propagation paths from 1 to 4 (sequentially adding the E-E, E-F, F-E, and F-F modes). Under these conditions, we evaluate the tracking accuracy of BP, the number of iterations required for convergence, and the average per-step execution time.

\subsubsection{Evaluation Metrics}
Across all MC trials in Experiments~II--IV,
we evaluate three key aspects: tracking accuracy via the average Optimal Subpattern Assignment (OSPA) distance~\cite{Schuhmacher2008OSPA}, algorithmic convergence via the average number of iterations required for BP message 
convergence (Avg. BP Iters), and computational efficiency via the average per-step execution 
time (Avg. Time) of each method.

\subsection{Results of Experiments}
The results of the four experiments are shown in Fig.~\ref{fig:mpda_bp_ame},
Fig.~\ref{fig:convergence_iterations}, Fig.~\ref{fig_BPMHT}, 
Table~\ref{tab:runtime_comparison}, and 
Table~\ref{tab:bp_single_column_compressed}.

\subsubsection{Experiment I}

Fig.~\ref{fig:mpda_bp_ame} shows the AME for both values of
$p_{\mathrm{d}}$ across the tested values of $\rho$.
At $\rho=5$~km, closely spaced targets induce strong association
ambiguity, resulting in the largest AME, particularly for
$p_{\mathrm{d}}=0.9$.
As $\rho$ increases, the association ambiguity decreases and
the AME falls rapidly, becoming numerically negligible for
$\rho\in\{20,25\}$~km.
These results confirm that the converged BP beliefs are highly
accurate under moderate or weak association ambiguity, whereas
noticeable approximation error may arise in the most ambiguous
case.

\begin{figure}[!htp]
\centering
\includegraphics[width=0.4\textwidth]{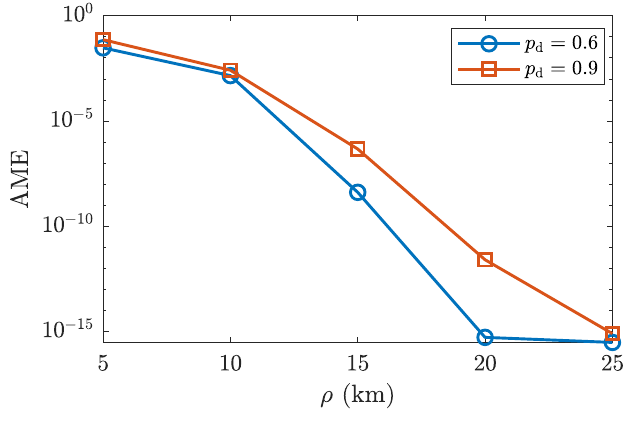}
\caption{AME of the converged BP beliefs relative to the exact marginal association probabilities under different values of $\rho$ and $p_d$.}
\label{fig:mpda_bp_ame}
\end{figure}

\subsubsection{Experiment II}
Fig.~\ref{fig:convergence_iterations} presents a histogram 
showing the number of BP iterations required to satisfy the 
convergence criterion $\|\boldsymbol{\mu}^{(n)} - 
\boldsymbol{\mu}^{(n-1)}\|_\infty < \delta$ across 500 
independent MC trials.
The horizontal axis denotes the 
time index $k$, the vertical axis represents the number of 
iterations, and the color intensity indicates the number of 
MC runs converging at each iteration count. 
As shown in Fig.~\ref{fig:convergence_iterations}, 
Algorithm~\ref{alg1} converges within 80 iterations across 
all time steps in all 500 MC trials, with the average number 
of iterations (black curve) remaining below 30, 
demonstrating that the number of BP iterations remains
moderate even in this dense scenario with $X_k=100$
and four propagation paths.

\begin{figure}[!htp]
\centering
\includegraphics[width=0.4\textwidth]{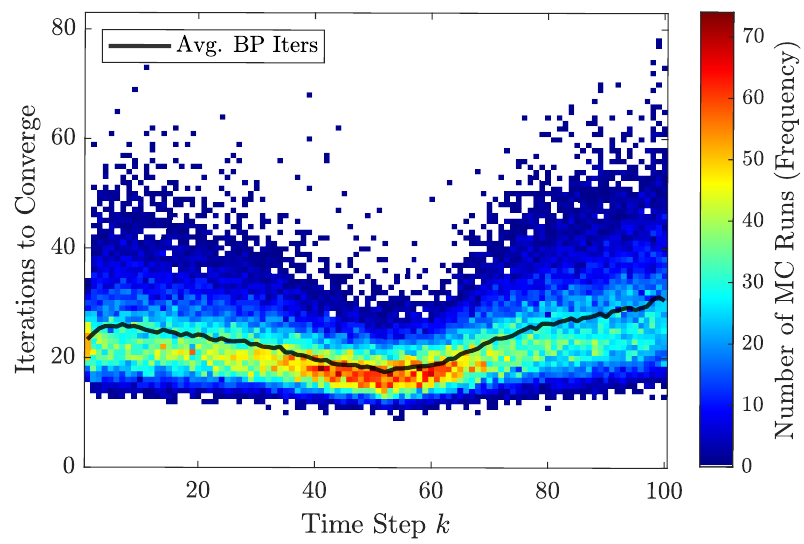}
\caption{
A 2D histogram showing the number of BP iterations 
required to satisfy the convergence criterion $e < \delta$ 
across 500 independent MC trials in a dense tracking scenario 
($\mathfrak{X}_k = 100$). Color intensity indicates the number 
of MC runs converging at each specific iteration count.}
\label{fig:convergence_iterations}
\end{figure}

\subsubsection{Experiment III}
Fig.~\ref{fig_BPMHT} compares the average OSPA distances of BP, MDMHT-1, and 
MDMHT-2~\cite{MHT2013Sathyan} under varying $p_{\text{d}}$ and 
average clutter number $\lambda_{c,k}$.
BP consistently achieves a lower average OSPA distance than both MD-MHT variants across all configurations.
As expected, lower $p_{\text{d}}$ or higher 
$\lambda_{c,k}$ increases the OSPA for the three methods.
Notably, a pronounced OSPA spike is observed around time step 50, corresponding to the moment 
when all targets intersect at the center, momentarily complicating the DA process.
Furthermore, Table~\ref{tab:runtime_comparison} summarizes the average per-step 
execution times. Although MDMHT-1 is marginally faster than BP, it 
comes at the cost of significantly degraded tracking accuracy. 
MDMHT-2 mitigates this accuracy loss via a multi-scan sliding window, but incurs higher computational cost and still falls short of BP in accuracy.
Consequently, BP achieves the most favorable accuracy-efficiency trade-off among 
all three methods.

\begin{figure}[!htp]
\centering
\includegraphics[width=0.45\textwidth]{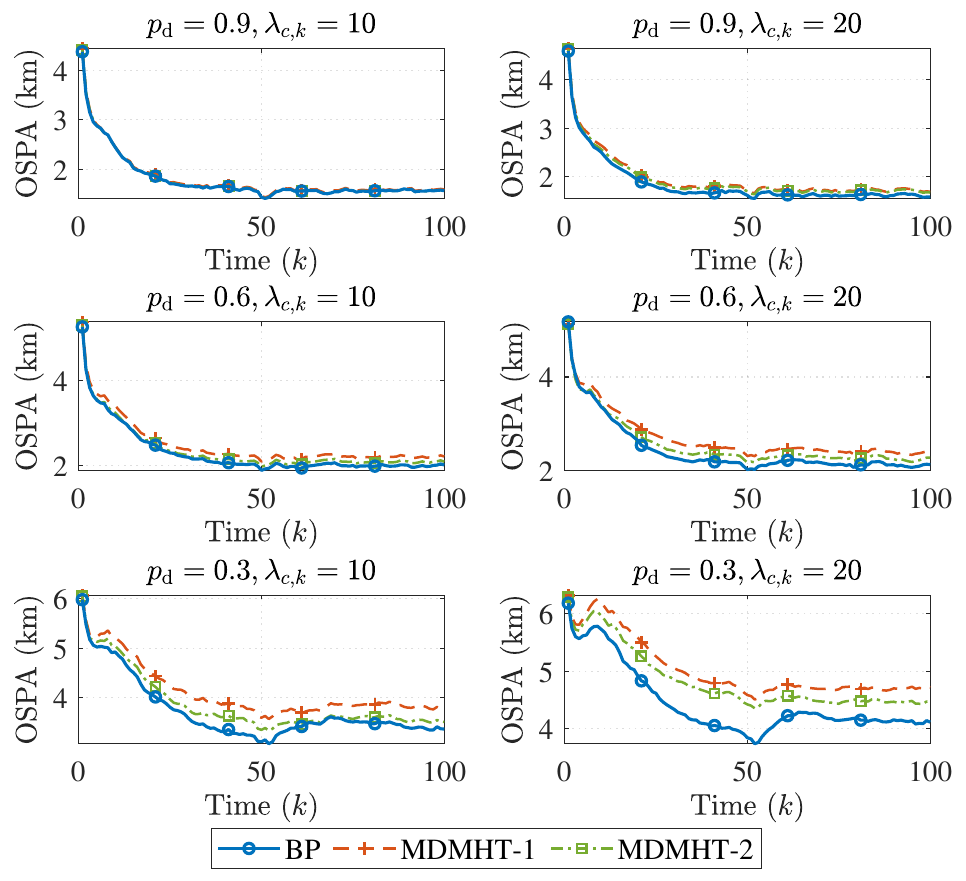}
\caption{Average OSPA of BP, MDMHT-1, and MDMHT-2 ($\mathfrak{X}_k = 3$) under different combinations of $p_{\text{d}}$ and $\lambda_{c,k}$.}
\label{fig_BPMHT}
\end{figure}

\begin{table}[htbp]
    \centering
    \caption{Average per-step execution time (ms) of each method 
    ($\mathfrak{X}_k = 3$, $\lambda_{c,k}^{(1)} = 10$, 
    $\lambda_{c,k}^{(2)} = 20$).}
    \label{tab:runtime_comparison}
    \setlength{\tabcolsep}{4.5pt}
    \footnotesize
    \begin{tabular}{l cc cc cc}
        \toprule
        \multirow{2}{*}{\textbf{Method}} 
        & \multicolumn{2}{c}{$p_{\text{d}}=0.9$} 
        & \multicolumn{2}{c}{$p_{\text{d}}=0.6$} 
        & \multicolumn{2}{c}{$p_{\text{d}}=0.3$} \\
        \cmidrule(lr){2-3} \cmidrule(lr){4-5} \cmidrule(lr){6-7}
        & $\lambda_{c,k}^{(1)}$ & $\lambda_{c,k}^{(2)}$ 
        & $\lambda_{c,k}^{(1)}$ & $\lambda_{c,k}^{(2)}$ 
        & $\lambda_{c,k}^{(1)}$ & $\lambda_{c,k}^{(2)}$ \\
        \midrule
        MDMHT-1 & 1.4 & 1.4 & 1.5 & 1.3 & 1.5 & 1.3 \\
        \textbf{BP} & 2.4 & 2.4 & 2.2 & 1.9 & 2.0 & 1.8 \\
        MDMHT-2 & 3.7 & 3.8 & 3.1 & 2.8 & 2.1 & 1.9 \\
        \bottomrule
    \end{tabular}
\end{table}

\subsubsection{Experiment IV}
Table~\ref{tab:bp_single_column_compressed} demonstrates the scalability of 
BP with respect to the number of targets~($\mathfrak{X}_k$) and propagation 
paths~($\mathfrak{M}_k$). 
Increasing $\mathfrak{X}_k$ moderately degrades tracking accuracy (higher OSPA) and increases both the number of BP iterations and the average per-step execution time, as additional targets and measurements enlarge the association problem and increase association ambiguity.
In contrast, increasing $\mathfrak{M}_k$ improves tracking accuracy by providing additional multipath information for target-state estimation, but at increased computational cost.
Specifically, additional paths increase the number of target--path combinations and measurements involved in DA, leading to more BP iterations and longer execution time.
As $\mathfrak{M}_k$ increases from 1 to 4, the average number of BP iterations increases from 54 to 94, while the average per-step execution time increases from 6.1~ms to 18.2~ms.
These results demonstrate the scalability of BP and its favorable trade-off between tracking accuracy and computational efficiency.

\begin{table}[htbp]
    \centering
    \caption{BP performance: OSPA, iterations, and per-step time under varying targets $\mathfrak{X}_k$ and paths $\mathfrak{M}_k$.}
    \label{tab:bp_single_column_compressed}
    \setlength{\tabcolsep}{2pt} 
    \begin{tabular}{l c c c c c c | c c c c}
        \toprule
        \multirow{2}{*}{\textbf{Metric}} & \multicolumn{6}{c|}{\textbf{Targets $\mathfrak{X}_k$}} & \multicolumn{4}{c}{\textbf{Paths $\mathfrak{M}_k$}} \\
        & 5 & 10 & 15 & 20 & 25 & 30 & 1 & 2 & 3 & 4 \\
        \midrule
        OSPA (km) & 1.84 & 1.93 & 2.19 & 2.34 & 2.51 & 2.73 & 2.78 & 2.22 & 1.76 & 1.63 \\
        Avg. BP Iters & 7 & 24 & 74 & 105 & 115 & 115 & 54 & 73 & 87 & 94 \\
        Time (ms) & 5.3 & 6.5 & 9.6 & 12.6 & 16.4 & 20.2 & 6.1 & 9.8 & 13.1 & 18.2 \\
        \bottomrule
    \end{tabular}
\end{table}

Overall, these results confirm the accuracy, convergence, and
scalability of BP in complex multipath MTT scenarios.

\section{Discussion}\label{sec:discussion}

So far, we have both theoretically proved and numerically demonstrated the convergence of BP for MPDA in MTT.
It is well known that BP has also been applied to extended object tracking~(EOT) for scalable DA, in which a single target may generate multiple measurements due to its spatial extent~\cite{Meyer2020Scalable}.
A natural question is then whether the convergence result developed in this work for MPDA can be applied to EOT by viewing the object extent as a kind of virtual multipath.
The answer is negative, for reasons explained below.

The measurement-generation mechanisms of MPDA and EOT are 
fundamentally different.
In EOT, multiple measurements arise from the continuous spatial extent of a single object under a unified observation model~\cite{Meyer2020Scalable}.
In MPDA, by contrast, a point target generates multiple measurements via multiple distinct propagation paths determined by the environment, and each path is governed by its own measurement function $h_{\tau,k}(\cdot)$.
Consequently, the feasible DA event spaces of MPDA and EOT 
are different, as formalized below.

In EOT, DA is described by binary variables 
$a_k^{i,j}\in\{0,1\}$, where $a_k^{i,j}=1$ indicates that 
measurement $j$ is generated by target $i$, reflecting a correspondence between targets and measurements.
Target $i$ may generate multiple measurements, subject to $\sum_j a_k^{i,j}\le l^{\max}$~\cite{Meyer2020Scalable}, 
where $l^{\max}$ is the maximum number of measurements 
that a single target can generate per scan.
For $\mathfrak{X}_k$ targets and $\mathfrak{Y}_k$ 
measurements, let $m_i\in[0,l^{\max}]$ denote the number 
of measurements assigned to target $i$.
Then the cardinality of the feasible EOT association-event space is
\begin{equation}\label{Eq:AkEOT}
|\mathcal{A}_{k}^{\text{EOT}}|
=
\sum_{\substack{0 \le m_i \le l^{\max} \\
\sum_i m_i \le \mathfrak{Y}_k}}
\frac{\mathfrak{Y}_k!}{m_1!\cdots m_{\mathfrak{X}_k}!\left(\mathfrak{Y}_k-\sum_{i=1}^{\mathfrak{X}_k}m_i\right)!}.
\end{equation}

In MPDA, however, the 
constraints~\eqref{Eq:FrameCon1}--\eqref{Eq:FrameCon2} 
enforce a three-way correspondence among targets, 
measurements, and propagation paths, treating each 
propagation path as distinct, so different permutations of path
labels correspond to different DA events.
For $\mathfrak{X}_k$ targets, $\mathfrak{M}_k$ paths, and 
$\mathfrak{Y}_k$ measurements, let $m_{i,\tau}\in\{0,1\}$ 
indicate whether path $\tau$ of target $i$ is assigned a 
measurement.
Then the cardinality of the feasible MPDA association-event space is
\begin{equation}\label{Eq:AkMPDA}
|\mathcal{A}_{k}^{\text{MPDA}}|
=
\sum_{\substack{m_{i,\tau}\in\{0,1\} \\
\sum_{i,\tau} m_{i,\tau} \le \mathfrak{Y}_k}}
\frac{\mathfrak{Y}_k!}{\left(\mathfrak{Y}_k-\sum_{i,\tau}m_{i,\tau}\right)!}.
\end{equation}

Hence, EOT and MPDA are defined over 
different feasible association-event spaces.
They coincide only in the degenerate case $\mathfrak{M}_k=1$ 
and $l^{\max}=1$, where both reduce to two-way DA.
This distinction is substantial rather than notational.
For example, when $\mathfrak{X}_k=2$, $\mathfrak{Y}_k=2$, 
$\mathfrak{M}_k=2$, and $l^{\max}=2$, \eqref{Eq:AkEOT} and 
\eqref{Eq:AkMPDA} give $|\mathcal{A}_{k}^{\mathrm{EOT}}|=9$ 
and $|\mathcal{A}_{k}^{\mathrm{MPDA}}|=21$, respectively.
Thus, even before BP is employed, the two models are defined 
over different association-event spaces and, in general, 
induce different association marginals and different 
posteriors.
This difference persists even when $\mathfrak{M}_k = l^{\max}$,
because in MPDA each path label is distinct, so permuting
path labels changes the DA event. Thus, assigning the same two measurements 
via paths $(\tau_1,\tau_2)$ and $(\tau_2,\tau_1)$ constitutes two different 
DA events, whereas in EOT only the subset of assigned measurements matters, 
making these two assignments the same DA event.
The reason is that EOT distinguishes only which subset of 
measurements is assigned to each target, whereas MPDA 
additionally distinguishes through which path each assigned 
measurement is received.

One may wonder whether EOT with at most $l^{\max}$ 
measurements per target can be rewritten as an MPDA 
problem by introducing $l^{\max}$ virtual paths per target.
This reduction fails for the following reasons.

Even if $l^{\max}$ virtual paths are introduced and share 
the same measurement function, they remain labeled paths in 
MPDA.
Consequently, a single EOT DA event in which target $i$ 
receives $m_i$ measurements from some subset is represented 
by $\frac{l^{\max}!}{(l^{\max}-m_i)!}$ distinct labeled 
path assignments, all corresponding to the same DA event.
The virtual-path model therefore over-counts EOT DA events 
by exactly this factor.

To compensate for this over-counting and recover the correct 
EOT posterior, each labeled path assignment must be assigned 
prior mass proportional to
\begin{equation}\label{Eq:SlotPrior}
\frac{(l^{\max}-m_i)!}{l^{\max}!}\,m_i!\,p(m_i\mid x_{i,k}).
\end{equation}
After summing over all $\frac{l^{\max}!}{(l^{\max}-m_i)!}$ 
equivalent labeled assignments, the total mass associated 
with a given measurement subset is proportional to 
$m_i!\,p(m_i\mid x_{i,k})$, which is precisely the target-wise 
counting weight used in the overcomplete EOT construction 
of~\cite{Meyer2020Scalable}.
This correction is mandatory --  omitting it inflates the 
posterior mass of every measurement subset by the 
over-counting factor, yielding a model that is no longer 
equivalent to EOT.

One might hope that a Poisson count model eliminates the 
need for this correction.
Under the Poisson assumption
$p(m_i\mid x_{i,k})=e^{-\gamma_i(x_{i,k})}\,\gamma_i(x_{i,k})^{m_i}/m_i!$,
where $\gamma_i(x_{i,k})>0$ denotes the expected
number of measurements generated by target~$i$,
giving $m_i!\,p(m_i\mid x_{i,k}) = e^{-\gamma_i(x_{i,k})}\,\gamma_i(x_{i,k})^{m_i}$.
However, this cancellation does not remove the 
path-permutation correction.
Substituting the Poisson assumption into~\eqref{Eq:SlotPrior} 
yields
\begin{equation}\label{Eq:PoissonSlotComp3}
\!\! \frac{(l^{\max}-m_i)!}{l^{\max}!}\,m_i!\,p(m_i\mid x_{i,k})
\!=\!
\frac{(l^{\max}-m_i)!}{l^{\max}!}\,e^{-\gamma_i}
\gamma_i^{m_i}.\!\!
\end{equation}
The residual factor $(l^{\max}-m_i)!\,/\,l^{\max}!$ survives 
because the virtual paths are exchangeable auxiliary labels 
rather than physical propagation paths, and no distributional 
assumption on $p(m_i\mid x_{i,k})$ can remove it.

Even granting the prior compensation above, the correction 
factor in~\eqref{Eq:SlotPrior} depends on the total count 
$m_i$ of measurements assigned to target $i$ across all its 
virtual paths and contains the counting correction 
proportional to $(l^{\max}-m_i)!\,m_i!$, thereby coupling 
all virtual paths of the same target.
Crucially, $(l^{\max}-m_i)!\,m_i!$ is a function of the 
aggregate count $m_i = \sum_{\tau} m_{i,\tau}$ across all 
virtual paths of target $i$ simultaneously; it cannot be 
decomposed into a product of per-path terms, and hence is 
non-separable across virtual paths.
This coupling is precisely the target-wise factor appearing 
in Eqs.~(31)--(34) of~\cite{Meyer2020Scalable}; see in 
particular Eq.~(33) of~\cite{Meyer2020Scalable}.
Because this factor is non-separable across virtual paths, 
the corrected virtual-path model yields the EOT BP message 
updates of~\cite{Meyer2020Scalable}, Eqs.~(36)--(45), 
rather than the MPDA BP message updates 
in~\eqref{Eq:muTk}--\eqref{Eq:muMk}, 
and therefore does not admit the MPDA factorization.

The cause of all incompatibilities above is the same 
fundamental difference.
In MPDA, each label indexes a distinct physical propagation 
path, so exchanging two labels changes the DA event itself.
In EOT, virtual path labels are exchangeable auxiliary 
variables, so different permutations of the occupied virtual 
paths represent the same DA event and require a target-wise 
prior correction to account for this over-counting, as given 
in~\eqref{Eq:SlotPrior}.

In summary, EOT is not a special case of 
MPDA.
Although EOT admits an exact overcomplete virtual-path 
reformulation, that reformulation still does not admit the MPDA factorization.
Hence, neither the MPDA loopy BP updates nor the 
convergence proof developed in this work extends directly 
to EOT, and a rigorous convergence analysis for loopy BP in 
EOT DA will be investigated in a companion article.

\section{Conclusion}\label{sec:conclusion}
We rigorously proved the convergence of BP for MPDA at each time step. 
Simulation results confirmed the accuracy of the converged
BP beliefs, the consistent convergence of the algorithm,
and a favorable accuracy--efficiency trade-off relative
to MD-MHT.

\bibliographystyle{ieeetr}  
\bibliography{autosam_TAES}

\end{document}